\documentclass[12pt,letterpaper]{article}
\usepackage{amsmath,amsfonts,amsthm,amssymb,stmaryrd,relsize}
\theoremstyle{plain}

\numberwithin{equation}{section}
\newtheorem{thm}{Theorem}[section]
\newtheorem{lem}[thm]{Lemma}
\newtheorem{cor}[thm]{Corollary}

\newenvironment{exam}[1]%  be sure to add \qquad\qedsymbol before \end{exam}
{\begin{flushleft}\textbf{Example #1}.\enspace}%
{\end{flushleft}}

\allowdisplaybreaks  % introduced 7.15.15

\newcommand{\complex}{{\mathbb C}}
\newcommand{\real}{{\mathbb R}}
\newcommand{\positive}{{\mathbb N}}

\newcommand{\ascript}{{\mathcal A}}
\newcommand{\bscript}{{\mathcal B}}
\newcommand{\cscript}{{\mathcal C}}
\newcommand{\dscript}{{\mathcal D}}
\newcommand{\escript}{{\mathcal E}}

\newcommand{\iscript}{{\mathcal I}}
\newcommand{\pscript}{{\mathcal P}}
\newcommand{\rscript}{{\mathcal R}}
\newcommand{\sscript}{{\mathcal S}}

\newcommand{\ctimes}{\mathrel{\mathlarger\cdot}}
\newcommand{\rmprob}{\mathrm{Prob\,}}
\newcommand{\rmtr}{\mathrm{tr\,}}
\newcommand{\rmsh}{\mathrm{Sh\,}}

\newcommand{\ab}[1]{\left|#1\right|}
\newcommand{\doubleab}[1]{\left|\left|#1\right|\right|}
\newcommand{\brac}[1]{\left\{#1\right\}}
\newcommand{\paren}[1]{\left(#1\right)}
\newcommand{\sqbrac}[1]{\left[#1\right]}
\newcommand{\elbows}[1]{{\left\langle#1\right\rangle}}
\newcommand{\ket}[1]{{\left|#1\right>}}
\newcommand{\bra}[1]{{\left<#1\right|}}

\begin{document}

\title{CONTEXTS IN QUANTUM\\
MEASUREMENT THEORY}
\author{Stan Gudder\\ Department of Mathematics\\
University of Denver\\ Denver, Colorado 80208\\
sgudder@du.edu}
\date{}
\maketitle

\hskip 1.2pc\textit{This paper is dedicated to my colleague and friend, Paul Busch,}

\hskip 6pc\textit{who investigated quantum theory with heart.}
\bigskip

\begin{abstract}
State transformations in quantum mechanics are described by completely positive maps which are constructed from quantum channels. We call a finest sharp quantum channel a context. The result of a measurement depends on the context under which it is performed. Each context provides a viewpoint of the quantum system being measured. This gives only a partial picture of the system which may be distorted and in order to obtain a total accurate picture, various contexts need to be employed. We first discuss some basic definitions and results concerning quantum channels. We briefly describe the relationship between this work and ontological models that form the basis for contextuality studies. We then consider properties of channels and contexts. For example, we show that the set of sharp channels can be given a natural partial order in which contexts are the smallest elements. We also study properties of channel maps. The last section considers mutually unbiased contexts. These are related to mutually unbiased bases which have a large current literature. Finally, we connect them to completely random channel maps.
\end{abstract}

\section{Basic Definitions and Results}  % Section 1
Quantum systems are usually described by operators on a complex separable Hilbert space $H$. Let $\bscript (H)$ denote the set of bounded linear operators on $H$ and $\sscript (H)$ the set of self-adjoint operators in $\bscript (H)$. For $A,B\in\bscript (H)$ we write $A\le B$ if
$\elbows{\phi ,A\phi}\le\elbows{\phi ,B\phi}$ for all $\phi\in H$ and if $A\ge 0$ we say that $A$ is \textit{positive}. We denote the set of positive operators by $\sscript ^+(H)$. If $0\le A\le I$ we call $A$ an \textit{effect} and denote the set of effects by $\escript (H)$. Effects describe
yes-no (two-valued) measurements that may be unsharp (fuzzy) \cite{bcl95,blm96,bgl97,cht70,gud982,hz12,kra83}. The sharp effects are given by projections satisfying $P^2=P$. We denote the set of projections on $H$ by $\pscript (H)$. If $P\in\pscript (H)$ is a one-dimensional projection onto the subspace of $H$ generated by a unit vector $\phi$, we write $P=P_\phi =\ket{\phi}\bra{\phi}$. We call $P_\phi$ (and $\phi$) a \textit{pure state}. If $\rho\in\sscript ^+(H)$ is of trace class with $\rmtr (\rho )=1$, we call $\rho$ a \textit{density operator} or
\textit{mixed state} and denote the set of density operators by $\dscript (H)$. If the system is described by a state $\rho$ and
$A\in\escript (H)$, then $\rmprob (A\mid\rho )=\rmtr (\rho A)$ is the \textit{probability that} $A$ \textit{occurs} (has value yes). In particular, for a pure state $P_\phi$ % from here
\begin{equation*}
\rmprob (A\mid P_\phi )=\rmtr (P_\phi A)=\elbows{\phi ,A\phi}
\end{equation*}
If $A\in\sscript (H)$, then $\elbows{\phi ,A\phi}$ is the \textit{expectation} of $A$ in the state $\phi$.

A \textit{quantum measurement} is described by a \textit{positive operator-valued measure} (POVM) \cite{bcl95,blm96,bgl97,cht70,hz12,kra83}. A POVM is a map $X$ from the Borel subsets $\bscript (\real )$ into $\escript (H)$ such that
$X (\emptyset )=0$, $X(\real )=I$ and if $\Delta _i\in\beta (\real )$, $i=1,2,\ldots$, satisfy $\Delta _i\cap\Delta _j=\emptyset$ for $i\ne j$, then
$X(\cup\Delta _i)=\sum X(\Delta _i)$ where the convergence of $\sum  X(\Delta _i)$ is in the strong operator topology. Quantum measurements are also called \textit{observables} and if $X(\Delta )\in\pscript (H)$ for all $\Delta\in\bscript (\real )$, then $X$ is a
\textit{sharp observable}. It follows from the spectral theorem that there is a one-to-one correspondence between sharp observables and elements of $\sscript (H)$. If $X$ is a POVM and $\rho\in\dscript (H)$ then the probability that $X$ has a value in $\Delta$ becomes
\begin{equation}                % equation (1.1)
\label{eq11}
\rmprob (\Delta\mid X,\rho )=\rmtr\sqbrac{\rho X(\Delta )}
\end{equation}

If $M_i\in\bscript (H)$ with $M_i\ne 0$ and $\sum M_i^*M_i=\sum M_iM_i^*=I$, we call $\ascript =\brac{M_i}$ a \textit{unital channel} with
\textit{branches} $M_i$ \cite{blm96,cho75,hz12,kra83,nc00}. We then have the \textit{channel map} $L_\ascript\colon\bscript (H)\to\bscript (H)$ given by
\begin{equation*}
L_\ascript (A)=\sum M_i^*AM_i
\end{equation*}
(Our terminology is a little different from what is frequently used, but we think it is more descriptive. What is usually called a channel, we call a channel map.) Notice that $L_\ascript$ is linear, $L_\ascript (I)=I$, $L_\ascript (A^*)=L_\ascript (A)^*$ and $A\le B$ implies that
$L_\ascript (A)\le L_\ascript (B)$. It follows that if $A\in\sscript (H)$ then $L_\ascript (A)\in\sscript (H)$, if $A\in\sscript ^+(H)$ then
$L_\ascript (A)\in\sscript ^*(H)$ and if $A\in\escript (H)$ then $L_\ascript (A)\in \escript (H)$. Also, if $A$ is trace finite then
$\rmtr\sqbrac{L_\ascript (A)}=\rmtr (A)$. Hence, if $\rho\in\dscript (H)$, then $L_\ascript (\rho )\in\dscript (H)$. Maps of the form $L_\ascript$ are called \textit{completely positive} and a restriction of $L_\ascript$ to $\dscript (H)$ is called a \textit{state transformation} \cite{bbrv01,kra83,nc00}. The simplest type of state transformation is $\Phi (\rho )=U^*\rho U$, where $U$ is a unitary operator. In general, if $P\in\pscript (H)$ then $L_\ascript (P)\notin\pscript (H)$ and $L_\ascript$ is not injective except in the unitary case. Also, if $AB=BA$, then
$L_\ascript (A)L_\ascript (B)\ne L_\ascript (B)L_\ascript (A)$ in general.

If $\ascript$ is a channel and $X\colon\bscript (\real )\to\escript (H)$ is a POVM, then it is easy to check that $L_\ascript\circ X$ is again a POVM. Two POVM's $X,Y$ \textit{coexist} if there exists a POVM $Z\colon\bscript (\real )\times\bscript (\real )\to\escript (H)$ whose marginals are $X$ and $Y$ \cite{bcl95,blm96,bgl97,hz12}. That is, $Z(\Delta\times\real )=X(\Delta )$ and $Z(\real\times\Delta )=Y(\Delta )$ for all
$\Delta\in\bscript (\real )$. In this case, we call $Z$ the \textit{joint measurement} for $X$ and $Y$. If $\ascript$ is a channel and $X,Y$ coexist with joint measurement $Z$, then it is easy to check that $L_\ascript\circ X$, $L_\ascript\circ Y$ coexist with joint measurement
$L_\ascript\circ Z$. Thus, although channel maps need not preserve commutativity, they do preserve coexistence.

If $\rho\in\dscript (H)$ is a state, $X\colon\bscript (\real )\to\escript (H)$ is a measurement and $\Phi\colon\dscript (H)\to\dscript (H)$ is a state transformation, the probability that $X$ has a value in $\Delta\in\bscript (\real )$ when $\Phi$ is performed is the generalization of \eqref{eq11} given by
\begin{equation}                % equation (1.2)
\label{eq12}
\rmprob (\Delta\mid\rho ,\Phi ,X)=\rmtr\sqbrac{\Phi (\rho )X(\Delta )}
\end{equation}
If $\Phi =L_\ascript$ where $\ascript =\brac{M_j}$ and $\brac{\phi _i}$ is an orthonormal basis for $H$, then \eqref{eq12} becomes
\begin{align}                % equation (1.3)
\label{eq13}
\rmprob (\Delta\mid\rho,\Phi ,X)&=\sum _j\rmtr\sqbrac{M_j^*\rho M_jX(\Delta )}\notag\\
   &=\sum _{i,j}\elbows{\phi _i,M_j^*\rho M_jX(\Delta )\phi _i}
\end{align}

A channel $\ascript =\brac{P_i}$ is \textit{sharp} if $P_i\in\pscript (H)$, $i=1,2,\ldots\,$. We then have that $\sum P_i=I$. It follows that $P_iP_j=0$ for $i\ne j$ and hence, $L_\ascript\circ L_\ascript =L_\ascript$. We denote the set of sharp channels by $\rmsh (H)$. A \textit{context} is a channel $\ascript =\brac{P_i}\in\rmsh (H)$ for which $P_i$ is a one-dimensional projection. We conclude that
$P_i=P_{\phi _i}$ where $\brac{\phi _i}$ is an orthonormal basis for $H$. Conversely, given an orthonormal basis $\brac{\phi _i}$ for $H$, we have the context $\brac{P_{\phi _i}}$. If $\ascript$ is a context, we call $L_\ascript$ the \textit{context map} for $\ascript$. We say that
$A\in\bscript (H)$ is \textit{measurable} with respect to a channel $\ascript =\brac{M_i}$ if $AM_i=M_iA$ for all $i=1,2,\ldots\,$. If $\ascript$ is sharp, it is clear that $L_\ascript (A)$ is measurable with respect to $\ascript$. This need not hold if $\ascript$ is unsharp. We leave the simple proof of the following theorem to the reader.

\begin{thm}    % Theorem 1.1
\label{thm11}
{\rm (a)}\enspace If $A$ is measurable with respect to $\ascript$, then $L_\ascript (A)=A$. Conversely, if $L_\ascript (A)=A$ and $\ascript$ is sharp, then $A$ is measurable with respect to $\ascript$.
{\rm (b)}\enspace If $\ascript =\brac{P_{\phi _i}}$ is a context, then the following statements hold.
{\rm (i)}\enspace $L_\ascript (A)=\sum\elbows{\phi _i,A\phi _i}P_{\phi _i}$ for every
$A\in\bscript (H)$.
{\rm (ii)}\enspace $L_\ascript (A)L_\ascript (B)=L_\ascript (B)L_\ascript (A)$ for every $A,B\in\bscript (H)$.
{\rm (iii)}\enspace $L_\ascript (A)$ has pure point spectrum with eigenvalues $\elbows{\phi _i,A\phi _i}$ and corresponding eigenvectors
$\phi _i$, $i=1,2,\ldots\,$.
\end{thm}

Since the branches of a context $\ascript$ have rank~1, $\ascript$ corresponds to a finest sharp measurement channel which we shall discuss in more detail in Section~2. A context provides a view of the physical system described by $H$. In general this is only a partial view that can be distorted so a complete picture would require various contexts. (We shall discuss mutually unbiased contexts in Section~3.) Thus, if
$A\in\sscript (H)$ represents a sharp observable, the $L_\ascript (A)$ is that observable from the viewpoint of context $\ascript$. In general, $L_\ascript (A)$ does not provide an accurate description unless $A$ is measurable with respect to $\ascript$. For example, by
Theorem~\ref{thm11}, $L_\ascript (A)L_\ascript (B)=L_\ascript (B)L_\ascript (A)$ for all $A,B\in\bscript (H)$ so a single context cannot describe quantum interference. In particular, it is easy to show that $A,B\in\sscript (H)$ with pure point spectra commute if and only if $A$ and $B$ are both measurable with respect to a single context. Other examples are: if $P\in\pscript (H)$ then $L_\ascript (P)$ need not be in $\pscript (H)$ and if $P_\phi$ is a pure state, then $L_\ascript (P_\phi )$ may be mixed.

Letting $2^\positive$ be the power set on $\positive =\brac{1,2,\ldots}$ we obtain the measurable space $(\positive ,2^\positive )$. Let
$\ascript =\brac{P_{\phi _i}}$ be a fixed context. Corresponding to a state $\rho\in\dscript (H)$, the operator
$L_\ascript (\rho )=\sum\elbows{\phi _i,\rho\phi _i}P_{\phi _i}$, provides a probability measure $\mu _{\ascript ,\rho}$ on $\positive$ given by
$\mu _{\ascript ,\rho}(i)=\elbows{\phi _i,\rho\phi _i}$. Corresponding to measurement $X$, we have
\begin{equation*}
L_\ascript (\Delta )=\sum\elbows{\phi _i,X(\Delta )\phi _i}P_{\phi _i}
\end{equation*}
This gives a random variable $f_{\ascript ,X(\Delta )}$ on $\positive$ given by
\begin{equation*}
f_{\ascript ,X(\Delta )}=\elbows{\phi _i,X(\Delta )\phi _i}
\end{equation*}
In fuzzy probability theory, $f_{\ascript ,X(\Delta )}$ is called a \textit{fuzzy event} and $\Delta\to f_{\ascript ,X(\Delta )}$ is a
\textit{fuzzy observable} \cite{gud981,gud982}. The probability that $X$ has a value in $\Delta$ according to context $\ascript$ becomes
\begin{align}                % equation (1.4)
\label{eq14}
\rmprob _\ascript (\Delta\mid\rho,X)&=\sum _i\mu _{\ascript ,\rho}(i)f_{\ascript ,X(\Delta )}(i)
     =\sum _i\elbows{\phi _i,\rho\phi _i}\elbows{\phi _i,X(\Delta )\phi _i}\notag\\
     &=\sum _i\elbows{\phi _i,\rho P_{\phi _i}X(\Delta )\phi _i}=\rmtr\sqbrac{L_\ascript (\rho )X(\Delta )}\notag\\
     &=\rmtr\sqbrac{\rho L_\ascript\paren{X(\Delta )}}=\rmtr\sqbrac{L_\ascript (\rho )L_\ascript\paren{X(\Delta )}}
 \end{align}
According to \eqref{eq11}, this does not agree with the usual quantum probability given by
\begin{equation*}
\rmprob (\Delta \mid\rho ,X)=\rmtr\sqbrac{\rho X(\Delta )}=\sum\elbows{\phi _i,\rho X(\Delta )\phi _i}
\end{equation*}
In fact, they agree only if $\rho$ or $X(\Delta )$ is measurable with respect to $\ascript$.

If $\Phi$ is a state transformation given by $\Phi (\rho )=\sum M_i^*\rho M_i$ and $\ascript =\brac{P_{\phi _i}}$ is a context, we define a random matrix on $\positive\times\positive$ by
\begin{equation*}
M_{\ascript ,\Phi (\rho )}(i,j)=\elbows{M_j\phi _i,\rho M_j\phi _i}=\elbows{\phi _i,M_j^*\rho M_j\phi _i}
\end{equation*}
We interpret $M_{\ascript ,\Phi (\rho )}(i,j)$ as the probability that $\rho$ traverses branch $j$ according to state $\phi _i$. We see that
\begin{align*}
\sum _iM_{\ascript ,\Phi (\rho )}(i,j)&=\rmtr (M_j^*\rho M_j)\\
\sum _jM_{\ascript ,\Phi (\rho )}(i,j)&=\mu _{\ascript ,\Phi (\rho )}(i)
\end{align*}
and of course $\sum _{i,j}M_{\ascript ,\Phi (\rho )}(i,j)=1$. As in \eqref{eq14}, according to $\ascript$, the probability that $X$ has a value in
$\Delta$ when the system is in the state $\Phi (\rho )$ becomes
\begin{align*}
\rmprob _\ascript (\Delta\mid\rho ,\Phi ,X)&=\sum _{i,j}M_{\ascript ,\Phi (\rho )}(i,j)\elbows{\phi _i,X(\Delta )\phi _i}\\
    &=\rmtr\sqbrac{L_\ascript\paren{X(\Delta )}\Phi (\rho )}=\rmtr\sqbrac{X(\Delta )L_\ascript\paren{\Phi (\rho )}}\\
    &=\rmtr\sqbrac{L_\ascript\paren{X(\Delta )}L_\ascript\paren{\Phi (\rho )}}
\end{align*}
Again, this is not the same as the quantum probability \eqref{eq12} and they agree only if $X(\Delta )$ or $\Phi (\rho )$ are measurable with respect to $\ascript$.

Corresponding to a collection of contexts $\Gamma =\brac{\ascript}$, we have the corresponding collection of measures
$\mu _{\ascript ,\rho}$, random variables $f_{\ascript ,X(\Delta )}$ and random matrices $M_{\ascript ,\Phi (\rho )}$ on the measurable space
$(\positive ,2^\positive )$. We call this set of elements an \textit{ontological model} \cite{gudap,lwe18} for the physical system. These models form the basis for contextuality studies in the current literature \cite{gudap}.

\section{Properties of Contexts} % Section 2
We denote the set of rank~1 projections on $H$ by $\pscript _1(H)$. The next result differentiates contexts from among the sharp channels.
\begin{lem}    % Lemma 1.2
\label{lem12}
If $\ascript\in\rmsh (H)$, then $\ascript$ is a context if and only if $L_\ascript (P)L_\ascript (Q)=L_\ascript (Q)L_\ascript (P)$ for every
$P,Q\in\pscript _1(H)$.
\end{lem}
\begin{proof}
By Theorem~\ref{thm11}, if $\ascript$ is a context, then $L_\ascript (P)L_\ascript (Q)=L_\ascript (Q)L_\ascript (P)$ for every
$P,Q\in\pscript _1(H)$. To prove the converse, suppose that $\ascript =\brac{P_i}\in\rmsh (H)$ and
$L_\ascript (P)L_\ascript (Q)=L_\ascript (Q)L_\ascript (P)$ for every $P,Q\in\pscript _1(H)$. Let $\phi ,\psi$ be unit vectors in the range of
$P_i$ for a particular $i$ and suppose that $\elbows{\phi ,\psi}\ne  0$. Letting $P=P_\phi$ and $Q=P_\psi$ and using the fact that
$P_iP_j=0$ for $j\ne i$ we have
\begin{align*}
\elbows{ \phi ,\psi}\ket{\phi}\bra{\psi}&=P_i\ket{\phi}\bra{\phi}P_i\ket{\psi}\bra{\psi}P_i=P_iPP_iQP_i\\
    &=\sum _iP_iPP_i\sum _jP_jQP_j=L_\ascript (P)L_\ascript (Q)\\
    &=L_\ascript (Q)L_\ascript (P)=P_iQP_iPP_i=\elbows{\psi ,\phi}\ket{\psi}\bra{\phi}
\end{align*}
Operate on $\psi$ with both sides to obtain
\begin{equation*}
\elbows{\phi ,\psi}\phi =\ab{\elbows{\psi ,\phi}}^2\psi
\end{equation*}
We conclude that $\psi =c\phi$ where $c\in\complex$. This implies that $P_i\in\pscript _1(H)$ for all $i$ so $\ascript$ is a context.
\end{proof}

The next result shows that contexts distinguish different operators.

\begin{lem}    % Lemma 2.2
\label{lem22}
If $A,B\in\bscript (H)$ and $L_\ascript (A)=L_\ascript (B)$ for every context $\ascript$, then $A=B$.
\end{lem}
\begin{proof}
Let $\phi _1\in H$ with $\doubleab{\phi _1}=1$ and let $\phi _1,\phi _2,\ldots$ be an orthonormal basis for $H$. Form the  context
$\ascript =\brac{P_{\phi _i}}$. Assuming that $L_\ascript (A)=L_\ascript (B)$ we have
\begin{equation*}
\sum\elbows{\phi _i,A\phi _i}P_{\phi _i}=\sum\elbows{\phi _i,B\phi _i}P_{\phi _i}
\end{equation*}
Multiplying by $P_{\phi _1}$ gives $\elbows{\phi _1,A\phi _1}P_{\phi _1}=\elbows{\phi _1,B\phi _1}P_{\phi _1}$ and it follows that
$\elbows{\phi _1,A\phi _1}=\elbows{\phi _1,B\phi _1}$. We conclude that $\elbows{\phi ,A\phi}=\elbows{\phi ,B\phi}$ for all $\phi\in H$. It follows from Proposition~1.21 \cite{hz12} that $A=B$.
\end{proof}

\begin{lem}    % Lemma 2.3
\label{lem23}
If $\ascript =\brac{P_{\phi _i}}$, $\bscript =\brac{P_{\psi _i}}$ are contexts and $A\in\bscript (H)$, then $L_\ascript (A)=L_\bscript (A)$ if and only if $\elbows{\phi _j,A\phi _j}=\elbows{\psi _k,A\psi _k}$ whenever $\elbows{\psi _k,\phi _j}\ne 0$.
\end{lem}
\begin{proof}
Let $P_i=P_{\phi _i}$, $Q_i=P_{\psi _i}$ and suppose that $L_\ascript (A)=L_\bscript (A)$. We then have
\begin{equation*}
\sum\elbows{\phi _i,A\phi _i}P_i=\sum\elbows{\psi _i,A\psi _i}Q_i
\end{equation*}
Multiplying by $P_j$ on the right and $Q_k$ on the left gives
\begin{equation*}
\elbows{\phi _j,A\phi _j}Q_kP_j=\elbows{\psi _k,A\psi _k}Q_kP_j
\end{equation*}
It follows that
\begin{equation*}
\elbows{\phi _j,A\phi _j}\elbows{\psi _k,\phi _j}\ket{\psi _k}\bra{\phi _j}=\elbows{\psi _k,A\psi _k}\elbows{\psi _k,\phi _j}\ket{\psi_k}\bra{\phi _j}
\end{equation*}
If $\elbows{\psi _k,\phi _j}\ne 0$ we conclude that $\elbows{\phi _j,A\phi _j}=\elbows{\psi _k,A\psi _k}$. The converse result from reversing these steps.
\end{proof}

If $\ascript =\brac{M_i}$, $\bscript =\brac{N_j}$ are channels, it is easy to check that $\ascript\bscript =\brac{M_iN_j}$ is again a channel. We can then define the channel map
\begin{equation*}
L_{\ascript\bscript}(A)=\sum _{i,j}(M_iN_j)^*AM_iN_j=\sum _{i,j}N_j^*M_i^*AM_iN_j=L_\ascript\sqbrac{L_\bscript (A)}
\end{equation*}
Hence, $L_{\ascript\bscript}=L_\ascript L_\bscript$. Notice that if $\ascript\in\rmsh (H)$, then $\ascript\ascript =\ascript$ and
$L_\ascript L_\ascript =L_\ascript$. Calling $\iscript =\brac{I}$ the \textit{identity channel} we have that
$\iscript\ascript =\ascript\iscript =\ascript$ and $L_\ascript L_\iscript =L_\iscript L_\ascript =L_\ascript$ for every channel $\ascript$. Of course, $\ascript\bscript\ne\bscript\ascript$ in general.

\begin{thm}    % Theorem 2.4
\label{thm24}
{\rm (a)}\enspace If $\ascript ,\bscript\in\rmsh (H)$, then $\ascript\bscript =\bscript\ascript$ if and only if $PQ=QP$ for all $P\in\ascript$,
$Q\in\bscript$.
{\rm (b)}\enspace If $\ascript ,\bscript$ are contexts, then $\ascript\bscript =\bscript\ascript$ if and only if $\ascript =\bscript$.
\end{thm}
\begin{proof}
(a)\enspace If $PQ=QP$ for all $P\in\ascript$, $Q\in\bscript$, then clearly $\ascript\bscript =\bscript\ascript$. Conversely, suppose that
$\ascript\bscript =\bscript\ascript$. If $P\in\ascript$ and $Q\in\bscript$, there exist $P_1\in\ascript$, $Q_1\in\bscript$ such that
$PQ=Q_1P_1$. If $P\ne P_1$, then we must have $P_1P=0$. Multiplying by $P$ on the right gives $PQP=0$. Hence, for every $\phi\in H$ we have
\begin{equation*}
\doubleab{QP\phi}^2=\elbows{QP\phi ,QP\phi}=\elbows{PQP\phi ,\phi}=0
\end{equation*}
It follows that $QP=0$ so $QP=(QP)^*=PQ$. If $P=P_1$, then multiplying by $P$ on the left gives $PQ=PQ_1P$. Hence, $PQ=(PQ)^*=QP$.
\newline (b)\enspace If $\ascript =\bscript$, the result is trivial. Suppose $\ascript\bscript =\bscript\ascript$ and $P_\phi\in\ascript$. Then there exists $P_\psi\in\bscript$ with $\elbows{\phi ,\psi}\ne 0$. By (a), $P_\phi P_\psi =P_\psi P_\phi$ so that
$\elbows{\phi ,\psi}\ket{\phi}\bra{\psi}=\elbows{\psi ,\phi}\ket{\psi}\bra{\phi}$. Operating on $\psi$ gives
\begin{equation*}
\elbows{\phi ,\psi}\phi =\ab{\elbows{\psi ,\phi}}^2\psi
\end{equation*}
Since $\elbows{\phi ,\psi}\ne 0$, we have that $\phi =c\psi$ with $\ab{c}=1$. Hence, $P_\phi =P_\psi$ so $P_\phi\in\bscript$. Similarly, $P_\psi\in\ascript$ for all $P_\psi\in\bscript$. Hence, $\ascript =\bscript$.
\end{proof}

\begin{cor}    % Corollary 2.5
\label{cor25}
If $\ascript ,\bscript\in\rmsh (H)$, then $\ascript\bscript\in\rmsh (H)$ if and only if $\ascript\bscript =\bscript\ascript$.
\end{cor}
\begin{proof}
By Theorem~\ref{thm24}(a), $\ascript\bscript =\bscript\ascript$ if and only if $PQ=QP$ for all $P\in\ascript$, $Q\in\bscript$. But
$PQ\in\pscript (H)$ if $PQ=(PQ)^*=QP$.
\end{proof}

\begin{cor}    % Corollary 2.6
\label{cor26}
If $\ascript$ and $\bscript$ are contexts and $L_\ascript (P)=L_\bscript (P)$ for all $P\in\ascript$, then $\ascript =\bscript$.
\end{cor}
\begin{proof}
If $P\in\ascript$, we have that $L_\bscript (P)=L_\ascript (P)=P$. We conclude that $PQ=QP$ for all $Q\in\bscript$ so
$\ascript\bscript =\bscript\ascript$. By Theorem~\ref{thm24}(b), $\ascript =\bscript$.
\end{proof}

If $\ascript ,\bscript\in\rmsh (H)$, we write $\ascript\le\bscript$ if for every $P\in\ascript$ there exists a $Q\in\bscript$ with $ P\le Q$. It is clear that $\ascript\le\bscript$, $\bscript\le\cscript$ implies that $\ascript\le\cscript$, $\ascript\le\ascript$ and $\ascript\le\iscript$. Moreover, if
$\ascript\le\bscript$ and $\bscript\le\ascript$ then for all $P\in\ascript$, there exists $Q\in\bscript$ and $P_1\in\ascript$ such that
$P\le Q\le P_1$. But then $P=P_1=Q$ and we conclude that $\ascript =\bscript$. It follows that $\rmsh (H)$ is a partially ordered set with largest element $\iscript$.

\begin{thm}    % Theorem 2.7
\label{thm27}
For $\ascript ,\bscript\in\rmsh (H)$ we have that $\ascript\le\bscript$ if and only if every $Q\in\bscript$ has the form $Q=\sum P_i$,
$P_i\in\ascript$.
\end{thm}
\begin{proof}
Assume that every $Q\in\bscript$ has the form $Q=\sum P_i$, $P_i\in\ascript$. Suppose $P\in\ascript$ and there does not exist a
$Q\in\bscript$ with $P\le Q$. Then for every $Q\in\bscript$ we have $Q=\sum P_i$, $P_i\in\ascript$ where $P_i\ne P$. But then $Q\perp P$ and we conclude that $I=\sum\brac{Q\colon Q\in\bscript}\perp P$. Hence, $P=0$ which is a contradiction so $\ascript\le\bscript$. Conversely, assume that $\ascript\le\bscript$. Let $Q\in\bscript$ and let $\brac{P_i}$ be the set of elements of $\ascript$ such that $P_i\le Q$. Then
$\sum P_i\le Q$ and Suppose that $\sum P_i<Q$. If $\brac{R_i}$ is the set of other projections in $\ascript$, we have $R_j\not\le Q$ so
$R_j\le Q_1\in\bscript$ where $Q_1\ne Q$. Since $Q_1\perp Q$, we have that $R_j\perp Q$ for every $j$. Since $R_j\le I-Q$ we have that
\begin{equation*}
I=\sum P_i+\sum P_j<Q+I-Q=I
\end{equation*}
which is a contradiction. We conclude that $\sum P_i=Q$.
\end{proof}

\begin{thm}    % Theorem 2.8
\label{thm28}
Let $\ascript ,\bscript\in\rmsh (H)$.
{\rm (a)}\enspace $\ascript\le\bscript$ if and only if $\ascript\bscript =\ascript$.
{\rm (b)}\enspace If $\ascript\le\bscript$, then $L_\ascript L_\bscript =L_\ascript$.
\end{thm}
\begin{proof}
(a)\enspace Suppose $\ascript\le\bscript$. Now
\begin{equation*}
\ascript\bscript=\brac{PQ\colon P\in\ascript ,Q\in\bscript ,PQ\ne 0}
\end{equation*}
Since $\ascript\le\bscript$, for every $P\in\ascript$ there exists a $Q\in\bscript$ with $P\le Q$ so $PQ=P$. Then for every $Q_1\in\bscript$ with $Q_1\ne Q$ we have that $Q_1\perp Q$ so $PQ_1=0$. Hence, $\ascript\bscript =\brac{P\colon P\in\ascript}=\ascript$. Conversely, suppose $\ascript\bscript =\ascript$. If $P\in\ascript$ then there exists $Q\in\bscript$ such that $PQ\ne 0$ because otherwise
$P=\sum\brac{PQ\colon Q\in\bscript}=0$ which is a contradiction. Since $\ascript\bscript =\ascript$, there exists a $P_1\in\ascript$ with $PQ=QP=P_1$. Now $P=P_1$ or $PP_1=0$. If $PP_1=0$, then $P_1=P_1PQ=0$ which is a contradiction. Therefore, $P=P_1$ so $PQ=P$ and $P\le Q$. We conclude that $\ascript\le\bscript$.
(b)\enspace If $\ascript\le\bscript$, then by (a) we have that $\ascript\bscript =\ascript$. Therefore,
\begin{equation*}
L_\ascript L_\bscript =L_\bscript L_\ascript =L_\ascript\qedhere
\end{equation*}
\end{proof}

The next result shows that contexts are the smallest sharp channels.

\begin{thm}    % Theorem 2.9
\label{thm29}
{\rm (a)}\enspace If $\ascript$ is a context and $\bscript\in\rmsh (H)$ with $\bscript\le\ascript$, then $\bscript =\ascript$.
{\rm (b)}\enspace Under the order $\le$, $\rmsh (H)$ is not a lattice.
\end{thm}
\begin{proof}
(a)\enspace Since $\bscript\le\ascript$ if $Q\in\bscript$ then there exists a $P\in\ascript$ such that $Q\le P$. Since $P\in\pscript _1(H)$, $Q=P$. Hence, $\bscript\subseteq\ascript$. Since $\sum\brac{Q\colon Q\in\bscript}=I$, we have that $\bscript =\ascript$.
(b)\enspace Let $\ascript$ and $\bscript$ be distinct contexts. If $\cscript\in\rmsh (H)$ satisfies $\cscript\le\ascript ,\bscript$ then by (a)
$\ascript =\cscript =\bscript$ which is a contradiction. Hence, the greatest lower bound $\ascript\wedge\bscript$ does not exist.
\end{proof}

\section{Mutually Unbiased Contexts} % Section 3
We have seen in Theorem~\ref{thm24}(b) that for contexts $\ascript$, $\bscript$ we have that $\ascript\bscript =\bscript\ascript$ if and only if
$\ascript =\bscript$. We shall later give an example in which $\ascript\ne\bscript$ but $L_\ascript L_\bscript =L_\bscript L_\ascript$. We first characterize the equality.

\begin{thm}    % Theorem 3.1
\label{thm31}
Let $\ascript =\brac{P_i}$, $\bscript =\brac{Q_i}$ be contexts with $P_i=P_{\phi _i}$, $Q_i=P_{\psi _i}$. Then
$L_\ascript L_\bscript =L_\bscript L_\ascript$ if and only if
\begin{equation}                % equation (3.1)
\label{eq31}
\sum _j\ab{\elbows{\phi _i,\psi _j}}^2\elbows{\psi _j,\phi _r}\elbows{\phi _s,\psi _j}=\delta _{rs}\ab{\elbows{\psi _k,\psi _r}}^2
\end{equation}
whenever $\elbows{\phi _i,\psi _k}\ne 0$
\end{thm}
\begin{proof}
Suppose $L_\ascript L_\bscript =L_\bscript L_\ascript$. Operating on $A=\ket{\phi _r}\bra{\phi _s}$ gives
\begin{align*}
L_\ascript\sqbrac{L_\bscript (A)}&=\sum _{i,j}P_iQ_j\ket{\phi _r}\bra{\phi _s}Q_jP_i\\
  &=\sum _{i,j}\ab{\elbows{\phi _i,\psi _j}}^2\elbows{\psi _j,\phi _r}\elbows{\phi _s,\psi _j}P_i
\end{align*}
and
\begin{align*}
L_\bscript\sqbrac{L_\ascript (A)}&=\sum _{i,j}Q_iP_j\ket{\phi _r}\bra{\phi _s}P_jQ_i
    =\delta _{rs}\sum _iQ_i\ket{\phi _r}\bra{\phi _r}Q_i\\
    &=\delta _{rs}\sum _i\ab{\elbows{\psi _i,\phi _r}}^2Q_i
\end{align*}
If $r\ne s$ we conclude that
\begin{equation*}
\sum _j\ab{\elbows{\phi _i,\psi _j}}^2\elbows{\psi _j,\phi _r}\elbows{\phi _s,\psi _j}=0
\end{equation*}
If $r=s$, we have that
\begin{equation*}
\sum _{i,j}\ab{\elbows{\phi_i,\psi _j}}^2\ab{\elbows{\psi _j,\phi _r}}^2P_i=\sum _i\ab{\elbows{\psi _i,\phi _r}}^2Q_i
\end{equation*}
Operating on $\psi _k$ gives
\begin{equation*}
\sum _{i,j}\ab{\elbows{\phi _i,\psi _j}}^2\ab{\elbows{\psi _j,\phi _r}}^2\elbows{\phi _i,\psi _k}\phi _i=\ab{\elbows{\psi _k,\phi _r}}^2\psi _k
\end{equation*}
Taking the inner product with $\phi _i$ we have that
\begin{equation*}
\sum _j\ab{\elbows{\phi _i,\psi _j}}^2\ab{\elbows{\psi _j\phi _r}}^2\elbows{\phi _i,\psi _k}=\ab{\elbows{\psi _k,\phi _r}}^2\elbows{\phi _i,\psi _k}
\end{equation*}
If $\elbows{\phi _i,\psi _k}\ne 0$, we obtain
\begin{equation*}
\sum _j\ab{\elbows{\phi _i,\psi _j}}^2\ab{\elbows{\psi _j,\phi _r}}^2=\ab{\elbows{\psi _k,\phi _r}}^2
\end{equation*}
This gives \eqref{eq31}. Conversely, if \eqref{eq31} holds we can work backwards with our previous equations to show that
\begin{equation*}
L_\ascript\sqbrac{L_\bscript (A)}=L_\bscript\sqbrac{L_\ascript (A)}
\end{equation*}
The result then follows by linearity.
\end{proof}

\begin{exam}{1}  % Example 1
In $H=\complex ^2$, let $\phi _1=(1,0)$, $\phi _2=(0,1)$, $\psi _1=\tfrac{1}{\sqrt{2}}(1,1)$, $\psi _2=\frac{1}{\sqrt{2}}(-1,1)$. It is easy to check that \eqref{eq31} holds. For instance, if $r=1$, $s=2$, $i=1$ we have
\begin{equation*}
\sum _{j=1}^2\ab{\elbows{\phi _1,\psi _j}}^2\elbows{\phi _j,\phi _1}\elbows{\phi _2,\psi _j}
   =\tfrac{1}{2}\ctimes\tfrac{1}{\sqrt{2}}\ctimes\tfrac{1}{2}+\tfrac{1}{2}\ctimes\paren{-\tfrac{1}{\sqrt{2}}}\ctimes\tfrac{1}{\sqrt{2}}=0
\end{equation*}
Moreover, since $\ab{\elbows{\phi _i,\psi _j}}^2=1/2$ for all $i,j$, we have that
\begin{equation*}
\sum _{j=1}^2\ab{\elbows{\phi _1,\psi _j}}^2\ab{\elbows{\psi _j,\phi _r}}^2=\tfrac{1}{2}\ctimes\tfrac{1}{2}+\tfrac{1}{2}\ctimes\tfrac{1}{2}
=\tfrac{1}{2}=\ab{\elbows{\psi _k,\phi _r}}^2
\end{equation*}
Letting $\ascript =\brac{P_{\phi _1},P_{\phi _2}}$, $\bscript =\brac{P_{\psi _1},P_{\psi _2}}$ we conclude from Theorem~\ref{thm31} that
$L_\ascript L_\bscript =L_\bscript L_\ascript$.\quad\qedsymbol
\end{exam}

\begin{exam}{2}  % Example 2
In $H=\complex ^2$, let $\phi _1=(1,0)$, $\phi _2=(0,1)$, $\psi _1=\tfrac{1}{2}(\sqrt{3},1)$, $\psi _2=\frac{1}{2}(-1,\sqrt{3})$. 
We now show that \eqref{eq31} does not hold. If $r=1$, $s=2$, $i=1$ we have
\begin{equation*}
\sum _{j=1}^2\ab{\elbows{\phi _1,\psi _j}}^2\elbows{\psi _j,\phi _1}\elbows{\phi _2\psi _j}
=\tfrac{3}{4}\ctimes\tfrac{\sqrt{3}}{2}\ctimes\tfrac{1}{2}+\tfrac{1}{4}\ctimes\paren{-\tfrac{1}{2}}\ctimes\tfrac{\sqrt{3}}{2}=\tfrac{\sqrt{3}}{8}\ne 0
\end{equation*}
Moreover, if $r=s=i=k=1$, we have
\begin{equation*}
\sum _{j=1}^2\ab{\elbows{\phi _1,\psi _j}}^2\ab{\elbows{\psi _j,\phi _1}}^2=\tfrac{9}{16}+\tfrac{1}{16}=\tfrac{5}{8}\ne\tfrac{3}{4}
  =\ab{\elbows{\psi _1,\phi _1}}^2
\end{equation*}
We conclude that $L_\ascript L_\bscript\ne L_\bscript L_\ascript$.\quad\qedsymbol
\end{exam}

For the rest of this section, we assume that $H$ is finite-dimensional with $\dim H=n$. We say that $a$ unit vector $\psi\in H$ is
\textit{unbiased} in an orthonormal basis $\brac{\phi _i}$ for $H$ if $\ab{\elbows{\psi ,\phi _i}}^2=\ab{\elbows{\psi ,\phi _j}}^2$ for all $i,j$. This terminology stems from the fact that the transition probabilities $\ab{\elbows{\psi ,\phi _i}}^2$ are the same for all $i=1,\ldots ,n$. Since
\begin{equation*}
n\ab{\elbows{\psi ,\phi _i}}^2=\sum _{i=1}^n\ab{\elbows{\psi ,\phi _i}}^2=\doubleab{\psi}^2=1
\end{equation*}
we conclude that $\ab{\elbows{\psi ,\phi _i}}^2=1/n$, $I=1,\ldots ,n$. We say that two orthonormal bases $\brac{\psi _i}$, $\brac{\phi _i}$ are
\textit{mutually unbiased} if each $\psi _i$ is unbiased in $\brac{\phi _i}$ so that $\ab{\elbows{\psi _i,\phi _i}}^2=1/n$ for all $i,j=1,\ldots ,n$
\cite{bbrv01,debz10,wf89}. The two bases in Example~1 are mutually unbiased. Defining the basis $\eta _1=\tfrac{1}{\sqrt{2}}(1,i)$,
$\eta _2=\tfrac{1}{\sqrt{2}}(1,-i)$, all three of these bases are mutually unbiased. These are called a \textit{maximal set} of mutually unbiased bases because the largest number of such bases in $\complex ^2$ is three \cite{bbrv01,debz10,wf89}.  Two contexts are \textit{mutually unbiased} if their corresponding bases are mutually unbiased. Also, an operator $A\in\bscript (H)$ is \textit{unbiased} in a context
$\ascript =\brac{P_{\phi _i}}$ if $\elbows{\phi _i,A\phi _i}=\elbows{\phi _j,A\phi _j}$ for all $i,j=1,2,\ldots ,n$. It follows that 
\begin{equation*}
\rmtr (A)=\sum _{i=1}^n\elbows{\phi _i,A\phi _i}=n\elbows{\phi _j,A\phi _j}
\end{equation*}
so that $\elbows{\phi _j,A\phi _j}=\rmtr (A)/n$ for $j=1,2,\ldots ,n$.

If $\brac{\phi _i}$, $\brac{\psi _i}$ are mutually unbiased bases, we have that
\begin{align*}
\sum _j\ab{\elbows{\phi _i,\psi _j}}^2&\elbows{\psi _j,\phi _r}\elbows{\phi _s,\psi _j}
  =\tfrac{1}{n}\sum _j\elbows{\psi _j,\phi _r}\elbows{\phi _s,\psi _j}\\
  &=\tfrac{1}{n}\elbows{\phi _s,\phi _r}=\delta _{rs}\tfrac{1}{n}=\delta _{rs}\ab{\elbows{\psi _k,\phi _r}}^2
\end{align*}
Hence, \eqref{eq31} holds and it follows from Theorem~\ref{thm31} that if $\ascript =\brac{P_{\phi _i}}$, $\bscript =\brac{P_{\psi _i}}$ are the corresponding contexts, then $L_\ascript L_\bscript =L_\bscript L_\ascript$

\begin{thm}    % Theorem 3.2
\label{thm32}
{\rm (a)}\enspace If $\ascript =\brac{P_{\phi _i}}$, $\bscript =\brac{P_{\psi _i}}$ are mutually unbiased contexts, then 
\begin{equation*}
L_\ascript\sqbrac{L_\bscript (A)}=L_\bscript\sqbrac{L_\ascript (A)}=\rmtr (A)I/n
\end{equation*}
for all $A\in\bscript (H)$.
{\rm (b)}\enspace If $L_\ascript\sqbrac{L_\bscript (P)}=I/n$ for all $P\in\pscript _1(H)$ then $\ascript$ and $\bscript$ are mutually unbiased.
\end{thm}
\begin{proof}
(a)\enspace This follows from
\begin{align*}
L_\ascript\sqbrac{L_\bscript (A)}&=\sum _{i,j}P_jQ_iAQ_iP_j=\sum _{i,j}\ab{\elbows{\phi _i,\psi _j}}^2\elbows{\psi _i,A\psi _i}P_i\\
   &=\tfrac{1}{n}\sum _{i,j}\elbows{\psi _i,A\psi _i}P_j=\rmtr (A)I/n
\end{align*}
(b)\enspace Assume that $L_\ascript\sqbrac{L_\bscript (P)}=I/n$ for all $P\in\pscript _1(H)$. As in the proof of (a) we have
\begin{equation*}
\sum _{i,j}\ab{\elbows{\phi _j,\psi _i}}^2\elbows{\psi _i,P\psi _i}P_j=I/n
\end{equation*}
Operating on $\phi _k$ gives
\begin{equation*}
\sum _i\ab{\elbows{\phi _k,\psi _i}}^2\elbows{\psi _i,P\psi _i}\phi _k=\tfrac{1}{n}\phi _k
\end{equation*}
Letting $P=Q_r$ gives $\ab{\elbows{\phi _k,\psi _r}}^2=1/n$ for all $r,k=1,2,\ldots ,n$. Hence, $\ascript$, $\bscript$ are mutually unbiased.
\end{proof}

We call the map $R\colon\bscript (H)\to\bscript (H)$ given by $R(A)=\rmtr (A)I/n$ the \textit{completely random channel map}. Notice that for
$\rho\in\dscript (H)$ we have $R(\rho )=I/n$ which is called the \textit{completely random state}.

\begin{cor}    % Corollary 3.3
\label{cor33}
The following statements are equivalent for contexts $\ascript =\brac{P_i}$, $\bscript =\brac{Q_i}$.
{\rm (a)}\enspace $\ascript$, $\bscript$ are mutually unbiased.
{\rm (b)}\enspace $L_\ascript L_\bscript =R$.
{\rm (c)}\enspace $\sum _iP_iQ_kP_i=I/n$, $k=1,2,\ldots ,n$.
{\rm (d)}\enspace $P_jQ_kP_j=P_j/n$, $j,k=1,2,\ldots ,n$.
\end{cor}
\begin{proof}
That (a) and (b) are equivalent follows from Theorem~\ref{thm32}. If (b) holds, then
\begin{equation*}
\sum _iP_iQ_kP_i=L_\ascript\sqbrac{L_\bscript (Q_k)}=R(Q_k)=I/n
\end{equation*}
Hence, (b) implies (c). If (c) holds, then multiplying by $P_j$ gives $P_jQ_kP_j=P_j/n$ so (c) implies (d). If (d) holds, then clearly (c) holds and we have that
\begin{align*}
L_\ascript\sqbrac{L_\bscript (A)}&=L_\ascript\paren{\sum _j\elbows{\psi _j,A\psi _j}Q_j}=\sum _j\elbows{\psi _j,A\psi _j}\sum _iP_iQ_jP_i\\
  &=\sum _j\elbows{\psi _j,A\psi _j}I/n=\rmtr (A)I/n=R(A)
\end{align*}
Hence, (d) implies (b). Another way to show this is that (d) implies that
\begin{equation*}
\ab{\elbows{\phi _j,\psi _k}}^2P_j=P_jQ_kP_j=P_j/n
\end{equation*}
so we have $\ab{\elbows{\phi _j,\psi _k}}^2=1/n$. Hence, (d) implies (a).
\end{proof}

If $\ascript =\brac{P_{\phi _i}}$, $\bscript =\brac{P_{\psi _i}}$ are mutually unbiased contexts, then we have seen that the channel
\begin{equation*}
\ascript\bscript =\brac{P_{\phi _i}P_{\psi _j}}=\brac{\elbows{\phi _i,\psi _j}\ket{\psi _j}\bra{\phi _i}}
\end{equation*}
satisfies $L_{\ascript\bscript}(A)=L_\ascript L_\bscript (A)=R(A)$ for all $A\in\bscript (H)$. We call the related channel
$\rscript =\brac{\tfrac{1}{\sqrt{n}}\ket{\psi _j}\bra{\phi _i}}$ a \textit{completely random channel}. It is easy to check that $\rscript$ and
$\ascript\bscript$ are equivalent in the sense that $L_\rscript (A)=R(A)$ for every $A\in\bscript (H)$.

\begin{lem}    % Lemma 3.4
\label{lem34}
{\rm (a)}\enspace The map $R\colon\bscript (H)\to\bscript (H)$ is the unique channel map satisfying $RL_\ascript =L_\ascript R=R$ for every channel $\ascript$.
{\rm (b)}\enspace An operator $A\in\bscript (H)$ is unbiased in a context $\ascript$ if and only if $L_\ascript (A)=R(A)$.
\end{lem}
\begin{proof}
(a)\enspace If $\ascript$ is a channel, we have for every $A\in\bscript (H)$ that
\begin{align*} 
R\sqbrac{L_\ascript (A)}&=\rmtr\sqbrac{L_\ascript (A)}I/n=\rmtr (A)I/n=R(A)\\
\intertext{and}
L_\ascript\sqbrac{R(A)}&=L_\ascript\sqbrac{\rmtr (A)I/n}=\tfrac{\rmtr (A)}{n}L_\ascript (I)=\rmtr (A)I/n=R(A)
\end{align*}
For uniqueness if $R_1$ and $R_2$ satisfy the given equation, then
\begin{equation*} 
R_1=R_1R_2=R_2R_1=R_2
\end{equation*}
(b)\enspace Let $\ascript =\brac{P_{\phi _i}}$ be a context. If $A\in\bscript (H)$ is unbiased in $\ascript$, then
$\elbows{\phi _i,A\phi _i}=\rmtr (A)/n$, $i=1,2,\ldots ,n$. Hence,
\begin{equation*} 
L_\ascript (A)=\sum\elbows{\phi _i,A\phi _i}P_i=\tfrac{\rmtr (A)}{n}I=R(A)
\end{equation*}
Conversely, suppose that $L_\ascript (A)=R(A)$. We then obtain
\begin{equation*} 
\sum\elbows{\phi _i,A\phi _i}P_i=\rmtr (A)I/n
\end{equation*}
But then $\elbows{\phi _i,A\phi _i}P_i=\tfrac{\rmtr (A)}{n}P_i$. Hence, $\elbows{\phi _i,A\phi _i}=\rmtr (A)/n$, $i=1,2,\ldots ,n$ so $A$ is unbiased in $\ascript$.
\end{proof}

Motivated by Lemma~\ref{lem34}(b), we say that $A\in\bscript (H)$ is \textit{unbiased} in a channel $\ascript$ if $L_\ascript (A)=R(A)$. In this case, if $\ascript =\brac{M_i}$ we have that
\begin{equation}                % equation (3.2)
\label{eq32}
\sum M_i^*AM_i=\rmtr (A)I/n
\end{equation}
We see that a sufficient condition for \eqref{eq32} to hold is that
\begin{equation*} 
M_i^*AM_i=\rmtr (A)M_i^*M_i/n,\qquad i=1,2,\ldots ,n
\end{equation*}
There are examples which show that the converse does not hold.

\begin{thm}    % Theorem 3.5
\label{thm35}
Let $\ascript =\brac{P_i}\in\rmsh (H)$ and let $\phi _{i,j}$ be an orthonormal basis for the range of $P_i$.
{\rm (a)}\enspace For $A\in\bscript (H)$ we have that
\begin{equation*} 
L_\ascript (A)=\sum _{i,j,k}\elbows{\phi _{i,j},A\phi _{i,j}}\ket{\phi _{i,j}}\bra{\phi _{i,k}}
\end{equation*}
{\rm (b)}\enspace $A$ is unbiased in $\ascript$ if and only if
\begin{equation}                % equation (3.3)
\label{eq33}
\elbows{\phi _{i,j},A\phi _{i,k}}=\tfrac{\rmtr (A)}{n}\delta _{j,k},\qquad i=1,2,\ldots ,n
\end{equation}
\end{thm}
\begin{proof}
(a)\enspace Since $P_i=\sum _j\ket{\phi _{i,j}}\bra{\phi _{i,j}}$ we have that 
\begin{align*} 
L_\ascript (A)&=\sum _{i,j}P_{\phi _{i,j}}A\sum _kP_{\phi _{i,k}}=\sum _{i,j,k}P_{\phi _{i,j}}AP_{\phi _{i,k}}\\
   &=\sum _{i,j,k}\elbows{\phi _{i,j},A\phi _{i,k}}\ket{\phi _{i,j}}\bra{\phi _{i,k}}
\end{align*}
(b)\enspace If $A$ is unbiased in $\ascript$, then $\sum P_iAP_i=\rmtr (A)I/n$. Hence, $P_iAP_i=\rmtr (A)P_i/n$ and we have that
\begin{equation*}
\elbows{\phi _{i,j},A\phi _{i,k}}=\tfrac{\rmtr (A)}{n}\elbows{\phi _{i,j},\phi _{i,k}}=\tfrac{\rmtr (A)}{n}\delta _{j,k}
\end{equation*}
Conversely, suppose that \eqref{eq33} holds. We then have that
\begin{align*} 
P_iAP_i&=\sum _{j,k}\elbows{\phi _{i,j},A\phi _{i,k}}\ket{\phi _{i,j}}\bra{\phi _{i,k}}\\
   &=\tfrac{\rmtr (A)}{n}\sum _j\ket{\phi _{i,j}}\bra{\phi _{i,j}}=\rmtr (A)P_i/n
\end{align*}
Hence,
\begin{equation*} 
L_\ascript (A)=\sum P_iAP_i=\tfrac{\rmtr (A)}{n}\sum P_i=\tfrac{\rmtr (A)}{n}I
\end{equation*}
We conclude that $A$ is unbiased in $\ascript$.
\end{proof}

We say that $A\in\bscript (H)$ is \textit{strongly unbiased} in a channel $\ascript$ if $A^m$ is unbiased in $\ascript$, $m=1,2,\ldots\,$. The importance of strongly unbiased is the following. If $A$ is non-degenerate, self-adjoint and strongly unbiased in a context $\ascript$, then the eigenvectors of $A$ are mutually unbiased in $\ascript$. The converse also holds. Of course, if a projection is unbiased in $\ascript$, then it is strongly unbiased in $\ascript$. If $\ascript$ and $\bscript$ are mutually unbiased contexts, then any operator that is measurable with respect to
$\ascript$ is strongly unbiased in $\bscript$. As with unbiased operators, sets of strongly unbiased operators are closed under scalar multiplication, addition and taking adjoints. The next example shows that the product of unbiased operators need not be unbiased.

\begin{exam}{3}  % Example 3
Let $\ascript$ be the context corresponding to the standard basis on $\complex ^2$. Representing operators by matrices in this basis, let
\begin{equation*}
A=\begin{bmatrix}0&1\\\noalign{\smallskip}1&0\end{bmatrix},\quad
B=\begin{bmatrix}0&i\\\noalign{\smallskip}-i&0\end{bmatrix}
\end{equation*}
Then $A$ and $B$ are self-adjoint and unbiased in $\ascript$, but
\begin{equation*}
AB=\begin{bmatrix}-i&0\\\noalign{\smallskip}0&i\end{bmatrix}
\end{equation*}
is not unbiased in $\ascript$.\quad\qedsymbol
\end{exam}

\begin{exam}{4}  % Example 4
In $\complex ^3$, let $\ascript$ be the context corresponding to the standard basis and let
\begin{equation*}
A=\begin{bmatrix}1&0&1\\\noalign{\smallskip}0&1&0\\\noalign{\smallskip}1&0&1\end{bmatrix}
\end{equation*}
Then $A$ is self-adjoint and unbiased in $\ascript$. However, $A$ is not strongly unbiased because
\begin{equation*}
A^2=\begin{bmatrix}2&0&2\\\noalign{\smallskip}0&1&0\\\noalign{\smallskip}2&0&2\end{bmatrix}\qquad\qedsymbol
\end{equation*}
\end{exam}

\begin{exam}{5}  % Example 5
The following is a general unbiased operator in the standard basis for $\complex ^2$:
\begin{equation*}
A=\begin{bmatrix}a&b\\\noalign{\smallskip}c&a\end{bmatrix}
\end{equation*}
We have that
\begin{equation*}
A^2=\begin{bmatrix}a^2+bc&2ab\\\noalign{\smallskip}2ac&a^2+bc\end{bmatrix}
\end{equation*}
so $A^2$ is again unbiased.\quad\qedsymbol
\end{exam}

The previous example suggests that unbiased operators in $\complex ^2$ are always strongly unbiased and this is our last result.

\begin{thm}    % Theorem 3.6
\label{thm36}
If $A\in\bscript (\complex ^2)$ is unbiased in a context $\ascript$, then it is strongly unbiased in $\ascript$.
\end{thm}
\begin{proof}
We can transform $\ascript$ to the context corresponding to the standard basis by a unitary transformation $U$. This transforms $A$ to
$UAU^*$. Since $A$ is unbiased in $\ascript$, $UAU^*$ will be unbiased in the standard basis. Thus, if we prove this result for matrices it will hold in general. Proceeding in this way, we shall show by induction on $n$ that if $A$ is a matrix that is unbiased in the standard basis for
$\complex ^2$, then $A^n$ is also. Letting
\begin{equation*}
A=\begin{bmatrix}a_{11}&a_{12}\\\noalign{\smallskip}a_{21}&a_{11}\end{bmatrix}
\end{equation*}
We use the notation
\begin{equation*}
A^n=\begin{bmatrix}a_{11}^n&a_{12}^n\\\noalign{\smallskip}a_{21}^n&a_{22}^n\end{bmatrix}
\end{equation*}
We know that the result holds for $n=1,2$. Suppose the result holds for $j=1,2,\ldots ,n$ so that $a_{11}^j=a_{22}^j$, $j=1,2,\ldots ,n$. We then obtain
\begin{align*} 
A^n&=\begin{bmatrix}a_{11}^n&a_{12}^n\\\noalign{\smallskip}a_{21}^n&a_{11}^n\end{bmatrix}
   =\begin{bmatrix}a_{11}&a_{12}\\\noalign{\smallskip}a_{22}&a_{11}\end{bmatrix}
  \begin{bmatrix}a_{11}^{n-1}&a_{12}^{n-1}\\\noalign{\smallskip}a_{21}^{n-1}&a_{11}^{n-1}\end{bmatrix}\\
   &=\begin{bmatrix}a_{11}a_{11}^{n-1}+a_{12}a_{21}^{n-1}&a_{11}a_{12}^{n-1}+a_{12}a_{11}^{n-1}\\
   \noalign{\smallskip}a_{12}a_{11}^{n-1}+a_{11}a_{21}^{n-1}&a_{21}a_{12}^{n-1}+a_{11}a_{11}^{n-1}\end{bmatrix}
\end{align*}
Since
\begin{equation*} 
a_{11}a_{11}^{n-1}+a_{12}a_{21}^{n-1}=a_{21}a_{12}^{n-1}+a_{11}a_{11}^{n-1}
\end{equation*}
we have that $a_{12}a_{21}^{n-1}=a_{21}a_{12}^{n-1}$. Therefore,
\begin{equation*} 
a_{12}a_{21}^n=a_{12}a_{21}a_{11}^{n-1}+a_{12}a_{11}a_{21}^{n-1}=a_{12}a_{21}a_{11}^{n-1}+a_{21}a_{11}a_{12}^{n-1}\\
   =a_{21}a_{12}^n
\end{equation*}
We then have that
\begin{align*} 
A^{n+1}&=\begin{bmatrix}a_{11}^{n+1}&a_{12}^{n+1}\\\noalign{\smallskip}a_{21}^{n+1}&a_{22}^{n+1}\end{bmatrix}
   =\begin{bmatrix}a_{11}&a_{12}\\\noalign{\smallskip}a_{21}&a_{11}\end{bmatrix}
  \begin{bmatrix}a_{11}^n&a_{12}^n\\\noalign{\smallskip}a_{21}^n&a_{11}^n\end{bmatrix}\\
   &=\begin{bmatrix}a_{11}a_{11}^n+a_{12}a_{21}^n&a_{11}a_{12}^n+a_{12}a_{11}^n\\\noalign{\smallskip}
   a_{21}a_{11}^n+a_{11}a_{21}^n&a_{21}a_{12}^n+a_{11}a_{11}^n\end{bmatrix}
\end{align*}
We conclude that
\begin{equation*}
a_{11}^{n+1}=a_{11}a_{11}^n+a_{12}a_{21}^n=a_{11}a_{11}^n+a_{21}a_{12}^n=a_{22}^{n+1}
\end{equation*}
Hence, $A^{n+1}$ is unbiased so the result holds by induction.
\end{proof}

\end{document}